\documentclass[12pt]{article}
\usepackage{amsfonts}
\usepackage{amssymb}
\usepackage{mathrsfs}
\usepackage{srcltx}
\usepackage{xcolor}
\usepackage{amsmath}
\usepackage{amsthm}
\usepackage{amstext}
\usepackage{amsopn}

\usepackage{subfigure}
\usepackage{tikz}
\usepackage{ulem}
\usepackage{graphicx}
\usepackage{bm}
\newtheorem{theorem}{Theorem}[section]
\newtheorem{lemma}[theorem]{Lemma}

\newtheorem{corollary}[theorem]{Corollary}

\theoremstyle{definition}
\newtheorem{definition}[theorem]{Definition}

\theoremstyle{remark}

\numberwithin{equation}{section}

\newcommand{\ba}{\begin{array}}
\newcommand{\ea}{\end{array}}

\newcommand{\la}{\lambda}

\newcommand{\ds}{\displaystyle}

\begin{document}
\date{}
\title{ \bf\large{Impact of resource distributions on the competition of  species in stream environment}}

\author{\small Tung D. Nguyen\footnote{Department of Mathematics, Texas A\&M University, College Station, TX 77843, USA}, Yixiang Wu\footnote{Corresponding author. Department of Mathematical Sciences, Middle Tennessee State University, Murfreesboro, Tennessee 37132, USA},  Tingting Tang\footnote{Department of Mathematics and Statistics, San Diego State University, San Diego, CA 92182, USA}, Amy Veprauskas\footnote{Department of Mathematics, University of Louisiana at Lafayette, Lafayette, LA 70501, USA}, \\ \small
Ying Zhou\footnote{Department of Mathematics, Lafayette College, Easton, PA 18042, USA},  Behzad Djafari Rouhani\footnote{Department of Mathematical Sciences, University of Texas at El Paso, El Paso, TX 79968, USA},  and Zhisheng Shuai\footnote{Department of Mathematics, University of Central Florida, Orlando, Florida 32816, USA}}

\maketitle

\begin{abstract}
Our earlier work in \cite{nguyen2022population} shows that concentrating resources on the upstream end tends to maximize the total biomass in a metapopulation model for a stream species.
In this paper, we continue our research direction by further considering a Lotka-Volterra competition patch model for two stream species. We show that the species whose resource allocations maximize the total biomass has the competitive advantage. 
\\[2mm]
   
\end{abstract}

\section{Introduction}
The impact of  resource distributions on the persistence of a single species has been studied extensively  (e.g., \cite{cantrell2004spatial, cantrell1989diffusive, berestycki2005analysis,lou2006effects,cantrell1998effects}).  Lou \cite{lou2006effects} observed that   if the dynamics of the species is modeled by a reaction-diffusion model with logistic type nonlinearity, the total biomass of the species may exceed the carrying capacity. Later, the ratio of total biomass over carrying capacity was studied both theoretically and experimentally  \cite{bai2016optimization, deangelis2016dispersal,  inoue2021unboundedness, liang2012dependence, zhang2015effects, zhang2017carrying}.  
 The distributions to maximize the growth rate \cite{cantrell1989diffusive, cantrell1991diffusive, lamboley2016properties,lou2006minimization} and the total biomass \cite{ding2010optimal,  mazari2020optimal, mazari2022optimisation,mazari2021fragmentation, nagahara2018maximization} have been to shown to be of  bang-bang type. Similar maximizing  total biomass problems have been studied for patch models with logistic growth and  random movement \cite{liang2021optimal,Nagahara2021}.

Resource allocations may also affect the interactions of multiple species \cite{he2013effects,  he2016globalII, he2016globalI, he2017global, lin2014global, gourley2005two, wei2021coexistence}. In a two-patch Lotka-Volterra competition model, it was shown that the species with a more heterogeneous distribution of resources will never lose the competition \cite{lin2014global,gourley2005two}. Similarly, in a two species reaction-diffusion competition model, it was shown that a species with a heterogeneous spatial distribution of resources will outcompete a species with homogeneously distributed resources \cite{he2016globalII}. Meanwhile, when both species have a heterogeneous distribution of resources, the slower dispersing species wins competition \cite{he2017global}.
For reaction-diffusion models it has also been shown that, when the carrying capacity is proportional to the growth rate, there is no optimal form of resource allocation \cite{hutson2003evolution}. Meanwhile, Mazari showed that the species whose resource allocation results in the largest total biomass  wins the competition when the diffusion rate is  large in a competition model with multiple species \cite{mazari2019trait}.

Our study is motivated by a series of recent works on metapopulation models in stream environments, where the individuals have both random movement and directed drift. In \cite{nguyen2022population}, we considered the impact of the distribution of resources on the persistence of a single stream species. In particular, we showed that to maximize the total biomass one should concentrate the resources on the upstream ends while to maximize the growth rate of the population one may need to concentrate the resources on the downstream ends. In \cite{Jiang2020BMB, Jiang-Lam-Lou2021}, the authors studied the joint impact of the diffusion rate, advection rate and network topology on the competition outcome of two stream species in a three patch Lotka-Volterra competition model whose patches are constructed as shown in Figure \ref{river}. When there are $n$ patches aligned along a line, the results in  \cite{chen2022,chen2022invasion,chen2021} showed that the magnitude of movement rates, the convexity of the drift rates and the population loss rate at the downstream end can significantly alter the competition outcome of two species. For works on reaction-diffusion-advection models for stream species, we refer the interested readers to  \cite{lam2016emergence,Lou14, Lou15, Lutscher06,pachepsky2005persistence,speirs2001population,vasilyeva2012flow,wang2019persistence, yan2022competition,MR4222368}  and the references therein.

 Motivated by the aforementioned studies, we consider the following two-species $n$-patch Lotka-Volterra competition model in stream environment: 
\begin{equation}\label{patch}
\begin{cases}
\ds\frac{du_i}{dt}=\ds\sum_{j=1}^{n} (d D_{ij}+qQ_{ij})u_j+r_iu_i\left(1-\frac{u_i+v_i}{k}\right), &i=1,\dots,n,\;\;t>0,\\
\ds\frac{dv_i}{dt}=\ds\sum_{j=1}^{n} (d D_{ij}+qQ_{ij})v_j+s_iv_i\left(1-\frac{u_i+v_i}{k}\right), &i=1,\dots,n,\;\;t>0,\\
\bm u(0)=\bm u_0\ge(\not\equiv)  \bm 0, \; \bm v(0)=\bm v_0\ge(\not\equiv) \bm 0.
\end{cases}
\end{equation}
Here the vectors $\bm u=(u_1, \dots, u_n)$ and $\bm v=(v_1, \dots, v_n)$ denote the densities of two competing stream species at each patch location. 
The nonnegative vectors $\bm r=(r_1, \dots, r_n)$ and $\bm s=(s_1, \dots, s_n)$ are the growth rates of $\bm u$ and $\bm v$, respectively, and the carrying capacity is assumed to be a positive constant $k$ for all the patches.  Two  $n\times n$ matrices $D=(D_{ij})$ and $Q=(Q_{ij})$ 
 represent the random movement pattern and   directed drift pattern of individuals respectively. For the three network configurations with $n=3$ in Figure \ref{river} (see \cite{Jiang2020BMB, Jiang-Lam-Lou2021}), the corresponding matrices  $D$ and $Q$ are as follows:
 \begin{itemize}
 \item Case 1:
\begin{equation}\label{DQ}
    D=\begin{bmatrix}
  -2 & 1 &1 \\
  1&-1& 0\\
  1&0 & -1 \\
 \end{bmatrix},\;\;Q=\begin{bmatrix}
  -2 & 0 &0 \\
  1&0& 0\\
  1&0 & 0\\
 \end{bmatrix};
\end{equation}

 \item Case 2:
\begin{equation}\label{DQ1}
    D=\begin{bmatrix}
  -1 & 0 &1 \\
  0&-1& 1\\
  1&1 & -2 \\
 \end{bmatrix},\;\;Q=\begin{bmatrix}
  -1 & 0 &0 \\
  0&-1& 0\\
  1&1 & 0\\
 \end{bmatrix};
\end{equation}

\item Case 3:
\begin{equation}\label{DQ2}
    D=\begin{bmatrix}
  -1 & 1 &0 \\
  1&-2& 1\\
  0&1 & -1 \\
 \end{bmatrix},\;\;Q=\begin{bmatrix}
  -1 & 0 &0 \\
  1&-1& 0\\
  0&1 & 0\\
 \end{bmatrix}.
\end{equation}
\end{itemize}

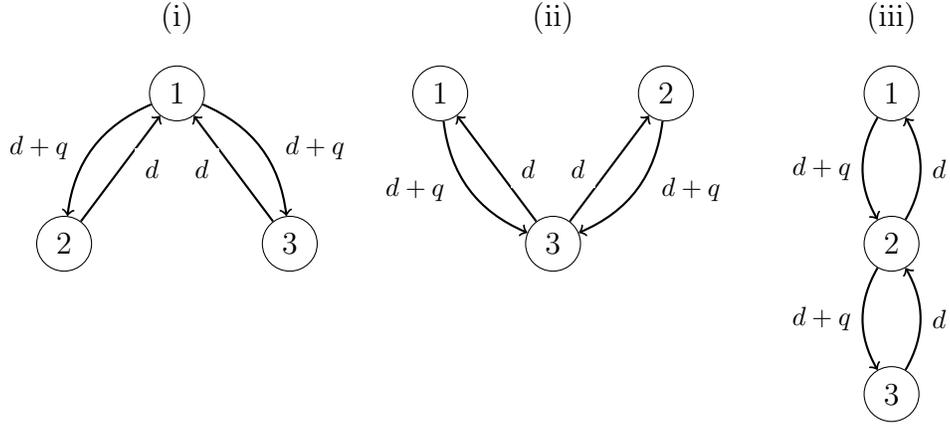
\begin{figure}[htbp]
\centering
\begin{tikzpicture}
\begin{scope}[every node/.style={draw}, node distance= 1.5 cm]
   
    \node[draw=white] (iii) at (13.5,1) {(iii)};
    \node[draw=white] (i) at (4,1) {(i)};
      \node[draw=white] (ii) at (9,1) {(ii)};
    
    \node[circle] (1) at (13.5,0) {$1$};
    \node[circle] (2) at (13.5,-2) {$2$};
    \node[circle] (3) at (13.5,-4) {$3$};
    
       \node[circle] (4) at (4,0) {$1$};
           \node[circle] (5) at (2.5,-2) {$2$};
    \node[circle] (6) at (5.5,-2) {$3$};

           \node[circle] (7) at (7.5,0) {$1$};
           \node[circle] (8) at (10.5,0) {$2$};
    \node[circle] (9) at (9,-2) {$3$};
\end{scope}
\begin{scope}[every node/.style={fill=white},
              every edge/.style={thick}]
    \draw[thick] [->](1) to [bend right] node[left=0.1] {{\footnotesize $d+q$}} (2);
    \draw[thick] [->](2) to [bend right] node[left=0.1] {{\footnotesize $d+q$}} (3);
    \draw[thick] [<-](1) to [bend left] node[right=0.1] {{\footnotesize $d$}} (2);
    \draw[thick] [<-](2) to [bend left] node[right=0.1] {{\footnotesize $d$}} (3);

        \draw[thick] [->](4) to [bend right] node[left=5] {{\footnotesize $d+q$}} (5);
        \draw[thick] [->](5) to node[right=5] {{\footnotesize $d$}} (4);
        \draw[thick] [->](4) to [bend left] node[right=5] {{\footnotesize $d+q$}} (6);
        \draw[thick] [->](6) to  node[left=5] {{\footnotesize $d$}} (4);
        
       \draw[thick] [->](7) to [bend right] node[left=5] {{\footnotesize $d+q$}} (9);
        \draw[thick] [->](9) to node[right=5] {{\footnotesize $d$}} (7);  
            \draw[thick] [->](8) to [bend left] node[right=5] {{\footnotesize $d+q$}} (9);
        \draw[thick] [->](9) to node[left=5] {{\footnotesize $d$}} (8);
\end{scope}
\end{tikzpicture}
\caption{A stream with three patches, where $d$ is the random movement rate and $q$ is the directed drift rate.  (i) Patch    1 is   the  upstream end and patches 2 and 3 are  the downstream ends. (ii) Patches  1 and 2  are  the  upstream ends and patch 3 is  the downstream end. (iii) Patch 1 is  the  upstream end and patch 3 is  the downstream end.}\label{river}
\end{figure}

Our objective is to determine how the distribution of resources for each species, as determined by $\boldsymbol{r}$ and $\boldsymbol{s}$, impact competitive outcomes for model \eqref{patch}.
We make the assumption that the resources are proportional to the growth rate in each patch and  the two species $\bm u$ and $\bm v$ have the same amount of resources, i.e. 
$$
\sum_{i=1}^n r_i=\sum_{i=1}^n s_i=r>0.
$$
We show that if the resources of species $\bm u$ are distributed  to maximize its biomass, i.e. all resources are distributed in the most upstream patches (see \cite{nguyen2022population}), while the resources of species $\bm v$ are not, then species $\bm u$  always wins the competition.  For example, for configuration (i) in Figure \ref{river}, such a distribution of species $\bm u$ corresponds to $\bm r=(r, 0, 0)$  while the distribution of species $\bm v$ satisfies $\bm s\neq (r,0,0)$.  


Our paper is organized as follows. In Section 2, we present some preliminary results which follow from existing theory. In Section 3, we consider the three-node stream networks shown in Figure \ref{river}. For each of these configurations we show that a species whose resources are distributed so that their total biomass is maximized in the absence of competition is able to out-compete a species whose resources are not optimally distributed. In Section 4, we extend these results to apply to $n$-patch stream networks.


\section{Preliminaries}

Let  $\bm w=(w_1, \dots, w_n)$ be a real vector. We write  $\bm w\gg \bm 0$ if $w_i>0$ for all $i=1, \dots, n$, and $\bm w>\bm 0$ if $\bm w\ge \bm 0$ but $\bm w\neq \bm 0$. Let $A=(a_{ij})_{n\times n}$ be a real square matrix. Let $\sigma(A)$ be the set of all eigenvalues of $A$, and $s(A)$ be the  {\it spectral bound} of $A$, i.e.
$$
s(A)=\max\{{\rm Re} \lambda: \lambda\in\sigma(A)\}.
$$
The matrix $A$ is called \emph{irreducible} if it cannot be placed into block upper triangular form by simultaneous row and column permutations and \emph{essentially nonnegative} if $a_{ij}\ge 0$ for all $1\le i, j\le n$ such that $i\neq j$. By the Perron-Frobenius Theorem, if $A$ is irreducible and essentially nonnegative, then $\lambda_1=s(A)$ is a simple eigenvalue of $A$. Moreover, $\lambda_1$ (called the \emph{principal eigenvalue} of $A$) is associated with an eigenvector whose components are all positive, which is the unique eigenvalue associated with a nonnegative eigenvector.

Before studying the two species competition model, we revisit the  following single species  meta-population model:
\begin{equation}\label{pat-s}
\begin{cases}
\ds\frac{du_i}{dt}=\ds\sum_{j=1}^{n}l_{ij} u_j+r_iu_i\left(1-\frac{u_i}{k_i}\right),&i=1,\dots,n,\;\;t>0,\\
\bm u(0)=\bm u_0>\bm0.
\end{cases}
\end{equation}
Here, $\bm u=(u_1, \dots, u_n)$ is the density of a meta-population living in $n$-patches; $\bm k=(k_1, \dots, k_n)$ is the carrying capacity; $\bm r=(r_1, \dots, r_n)$ is the growth rate.   
The coefficients $\ell_{ij}\ge 0$ denote the movement rate of the individuals from patch $j$ to patch $i$ for $1\le i,j\le n$ and $i\not= j$; 
$l_{ii}=-\sum_{j\neq i} l_{ji}$ is the total movement rate out from patch $i$. Then the $n\times n$ \textit{connection matrix} $L$ is of the form 
 \begin{equation}\label{eq-L}
L:=\begin{pmatrix}
-\sum_j \ell_{j1} & \ell_{12} & \cdots & \ell_{1n}\\
\ell_{21} & -\sum_j \ell_{j2} & \cdots & \ell_{2n}\\
\vdots & \vdots & \ddots & \vdots\\
\ell_{n1} & \ell_{n2} & \cdots & -\sum_j \ell_{jn}
\end{pmatrix}.
\end{equation}
 It is easy to see that $(1,1, \dots, 1)$ is a left eigenvector of $L$ corresponding to eigenvalue $0$. We always assume that  $L$ is irreducible. By the Perron-Frobenius Theorem, 0 is the principal eigenvalue of $L$.

We can associate $L$ with  a weighted, directed graph (digraph) $\mathcal{G}$ consisting of  $n$ nodes (each node $i$ in $\mathcal{G}$ corresponds to patch $i$). In $\mathcal{G}$, there is a directed edge (arc) from node $j$ to node $i$ if and only if $\ell_{ij}>0$. The couple $(\mathcal{G}, L)$ is called the \emph{movement network} associated with \eqref{pat-s}.

The global dynamics of \eqref{pat-s} are well-known:
 
\begin{lemma}[\cite{cosner1996variability,li2010global,Lu1993,takeuchi1996global}]\label{DS-single}
Suppose that $L$ is essentially nonnegative and irreducible matrix that is defined in \eqref{eq-L}. If $\bm r>\bm 0$ and  $\bm k\gg\bm 0$, then model \eqref{pat-s} has a unique positive equilibrium, which is globally asymptotically stable.
\end{lemma}

By Lemma \ref{DS-single}, model \eqref{patch} has two semitrivial equilibria $E_1:=(\bm u^*, \bm 0)$ and $E_2:=(\bm 0, \bm v^*)$.  By the well-known monotone dynamical system theory 
\cite{hess,hsu1996competitive,lam2016remark,smith2008monotone}, the global dynamics of \eqref{patch} is closely related to the local properties of its equilibria. 
Denote $X=\mathbb{R}_+^n\times\mathbb{R}_+^n$. Let $\le_K$ be the order in $X$ induced by the cone $K=\mathbb{R}_+^n\times\{-\mathbb{R}_+^n\}$. Then if $\bm x=(\bar {\bm u}, \bar {\bm v}), \bm y=(\tilde  {\bm u}, \tilde {\bm v})\in X$, we write $\bm x\le_K \bm y$ if $\bar  {\bm u}\le \tilde  {\bm u}$ and $\bar {\bm v}\ge \tilde  {\bm v}$; $\bm x<_K \bm y$ if $\bm x\le_K \bm y$ and $\bm x\neq \bm y$. 
We utilize the following result later  (this result is proved in \cite{smith2008monotone} for the case $n=2$ first but it holds for any $n\ge 2$  \cite[Page 70]{smith2008monotone}):
\begin{lemma}[{\cite[Theorem 4.4.2]{smith2008monotone}}]\label{M1}
Suppose that  $E_1$ is linearly unstable. Then one of the following holds:
\begin{itemize}
    \item[{\rm (i)}] $E_2$ attracts all solutions with initial data $(\bm u_0,\bm v_0)\in X$ satisfying $\bm v_0>0$. 
    In this case, $E_2$ is linearly stable or  neutrally stable; 
    \item[{\rm (ii)}] There exists a positive equilibrium $E$ satisfying $E_2\ll_K E\ll_K E_1$ such that $E$ attracts all solutions with initial data $(\bm u_0, \bm v_0)\in X$ satisfying $E\le_K (\bm u_0, \bm v_0)<_K E_1$. 
\end{itemize}
\end{lemma}

By Lemma \ref{M1}, if  $E_2$ is linearly unstable and the model has no positive equilibrium, then $E_1$ is globally attractive.  
It is easy to see that the stability of  $E_1$   is determined by the sign of  $\la_1(\bm{s}, \bm u^*)$, which is the principal eigenvalue of the matrix $dD+qQ+\text{diag}(s_i(1-u^*_i/k))$: if $\la_1(\bm{s}, \bm u^*)<0$, $E_1$ is locally asymptotically stable; if  $\la_1(\bm{s}, \bm u^*)>0$, $E_1$ is unstable;  if  $\la_1(\bm{s}, \bm u^*)=0$, $E_1$ is linearly neutrally stable.  Similarly, the local stability of $E_2$ is determined by the sign of $\la_1(\bm{r}, \bm v^*)$, which is the principal eigenvalue of the matrix $dD+qQ+\text{diag}(r_i(1-v^*_i/k))$.

\section{Stream networks of three nodes}\label{sec3nodes}
In this section, we consider model \eqref{patch} for the three-node stream networks shown in Figure \ref{river}. Here we provide detailed analysis for configuration (i), with analogous results for configurations (ii) and (iii) provided in the appendix. Our results state that a species whose resources are concentrated on the upstream end will have the competitive advantage.  In particular, for configuration (i) we show that if $\bm r=(r, 0, 0)$ and $\bm s\neq \bm r$, then species $\bm u$ always wins the competition. 

We first prove the following lemma which is used to show that model \eqref{patch} does not have a positive equilibrium.

\begin{lemma}\label{lemma_sign}
    Suppose that $D$ and $Q$ are given by \eqref{DQ}. Let $\bm r=(r, 0, 0)$ and $\bm s=(s_1, s_2, s_3)>\bm 0$ with $\bm s\neq \bm r$ and $\sum_{i=1}^3 s_i=r>0$. If $(\bm u, \bm v)$ is a positive equilibrium of \eqref{patch}, then $u_1+v_1<k$, $u_2+v_2>k$ and $u_3+v_3>k$.
\end{lemma}
\begin{proof}
Suppose that $(\bm u, \bm v)$ is a positive equilibrium of \eqref{patch}. Then $(\bm u, \bm v)$ satisfies 
\begin{equation}\label{pss}
\begin{cases}
\ds 0=\ds\sum_{j=1}^{3} (d D_{ij}+qQ_{ij})u_j+r_iu_i\left(1-\frac{u_i+v_i}{k}\right), &i=1,2,3,\\
\ds 0=\ds\sum_{j=1}^{3} (d D_{ij}+qQ_{ij})v_j+s_iv_i\left(1-\frac{u_i+v_i}{k}\right), &i=1,2,3.
\end{cases}
\end{equation}
Let $w_i:=u_i+v_i$ for $i=1, 2, 3$. Adding each corresponding pair of equations above, we have
\begin{eqnarray}
&&0=-(d+q)w_1+dw_2-(d+q)w_1+dw_3+(r_1u_1+s_1v_1)\left(1-\frac{w_1}{k} \right), \nonumber\\
&&0=(d+q)w_1-dw_2+(r_2u_2+s_2v_2)\left(1-\frac{w_2}{k} \right), \label{eq:sum_u_v}\\
&&0=(d+q)w_1-dw_3+(r_3u_3+s_3v_3)\left(1-\frac{w_3}{k} \right). \nonumber
\end{eqnarray}
We prove that $w_2=u_2+v_2 >k$ by contradiction.  Assume to the contrary that $w_2\le k$. Then by the second equation of \eqref{eq:sum_u_v}, we have  
\begin{equation*}
(d+q)w_1-dw_2 =  -(r_2u_2+s_2v_2)\left(1-\frac{w_2}{k} \right)\le 0.
\end{equation*} 
This implies
$$
w_1\le \frac{d}{d+q} w_2\le \frac{d}{d+q}k<k.
$$
By $r_1=r>0$ and the first equation of \eqref{eq:sum_u_v}, we have 
\[
(d+q)w_1 - dw_3 = -(d+q)w_1+dw_2 + (r_1u_1+s_1v_1)\left(1-\frac{w_1}{k} \right).
\]
Since we have shown that $(d+q)w_1 - dw_2 \leq 0$ and $w_1 <k$, we must have $(d+q)w_1-dw_3>0$. Thus
\[
w_3 < \frac{d+q}{d} w_1 \leq \frac{d+q}{d}\frac{d}{d+q}k = k.
\]
This further implies
\[
(d+q)w_1 - dw_3 +(r_3u_3+s_3v_3)\left(1-\frac{w_3}{k} \right) > 0,
\]
which contradicts the third equation of \eqref{eq:sum_u_v}. Therefore, we must have $u_2+v_2>k$. Similarly, we have $u_3+v_3>k$. 

Since $\bm s\neq\bm r$, either $s_2\neq 0$ or $s_3\neq 0$. Without loss of generality, say $s_2\neq 0$. Then by the second equation of \eqref{eq:sum_u_v} and $w_2>k$, $(d+q)w_1-dw_2>0$. By the third equation of \eqref{eq:sum_u_v} and $w_3>k$, we have $(d+q)w_1-dw_3\ge 0$. Finally by the first equation of \eqref{eq:sum_u_v} and $s_1>0$, we have $w_1=u_1+v_1<k$.
\end{proof}

Next, in Lemma \ref{lemma_nonexistence} we make use of the following well-known result (e.g., see \cite[ Corollary 2.1.5]{berman1994nonnegative})
to prove the non-existence of a positive equilibrium.

\begin{lemma}\label{lemma_sp}
Suppose that $P$ and $Q$ are $n\times n$ real-valued matrices, $P$ is essentially nonnegative, $Q$ is nonnegative and nonzero, and
$P+Q$ is irreducible. Then, $s(P+Q)>s(P)$.
\end{lemma}

\begin{lemma}\label{lemma_nonexistence}
    Suppose that $D$ and $Q$ are given by \eqref{DQ}. Let $\bm r=(r, 0, 0)$ and $\bm s=(s_1, s_2, s_3)>\bm 0$ with $\bm s\neq \bm r$ and $\sum_{i=1}^3 s_i=r>0$. Then model \eqref{patch} has no positive equilibrium. 
\end{lemma}
\begin{proof}
    Suppose to the contrary that $(\bm u, \bm v)$ is a positive equilibrium of \eqref{patch}. Then $(\bm u, \bm v)$ satisfies \eqref{pss}.
By the first equation of \eqref{pss}, $\bm u$ is a positive eigenvector of matrix $M_1:=dD+qQ+\text{diag}(r_i(1-(u_i+v_i)/k))$ corresponding with eigenvalue 0. By the Perron-Frobenius Theorem, we must have $s(M_1)=0$. Similarly,  $\bm v$ is a positive eigenvector of matrix $M_2:=dD+qQ+\text{diag}(s_i(1-(u_i+v_i)/k))$ corresponding with eigenvalue 0 and $s(M_2)=0$.
By the assumptions on $\bm r$ and $\bm s$ and Lemma \ref{lemma_sign}, we have 
$$
r_1\left(1-\frac{u_1+v_1}{k}\right)>s_1\left(1-\frac{u_1+v_1}{k}\right)
$$
and 
$$
0=r_i\left(1-\frac{u_i+v_i}{k}\right)\ge s_i\left(1-\frac{u_i+v_i}{k}\right), \ i=2, 3.
$$
Therefore, by Lemma \ref{lemma_sp}, we must have $s(M_1)>s(M_2)$, which is a contradiction. This proves the result. 
\end{proof}

In the following two lemmas, we show that the semitrivial equilibrium $E_2$ is always unstable.
\begin{lemma}\label{lemma_sign1}
    Suppose that $D$ and $Q$ are given by \eqref{DQ}. Let  $\bm s=(s_1, s_2, s_3)\ge\bm 0$ with $s_2>0$ or $s_3>0$. Then the semitrivial equilibrium $E_2=(\bm 0, \bm v^*)$ satisfies $v^*_1<k$.
\end{lemma}

\begin{proof}
We observe that $\bm v^*$ must satisfy
\[
0= \sum_{j=1}^3(dD_{ij}+qQ_{ij})v^*_j + s_iv_i^*\left(1-\frac{v_i^*}{k}\right), \ \ i=1, 2, 3.
\]
That is
\begin{eqnarray}
&&0=-(d+q)v_1^*+dv_2^*-(d+q)v^*_1+dv^*_3+s_1v^*_1\left(1-\frac{v^*_1}{k} \right), \nonumber\\
&&0=(d+q)v^*_1-dv^*_2+s_2v^*_2\left(1-\frac{v^*_2}{k} \right), \label{eq:sum_u_v1}\\
&&0=(d+q)v^*_1-dv^*_3+s_3v^*_3\left(1-\frac{v^*_3}{k} \right). \nonumber
\end{eqnarray}
Assume to the contrary that $v_1^*\ge k$. By the first equation of \eqref{eq:sum_u_v1}, either $(d+q)v^*_1-dv^*_2\le 0$ or $(d+q)v^*_1-dv^*_3\le 0$. Without loss of generality, say $s_2>0$. If $(d+q)v^*_1-dv^*_2\le 0$, then by the second equation of \eqref{eq:sum_u_v1}, we have $v^*_2\le k$. This implies 
\begin{equation}\label{v2ss}
(d+q)v^*_1-dv^*_2\ge(d+q)k-dk>0,
\end{equation}
which is a contradiction. Hence, $(d+q)v^*_1-dv^*_3\le 0$. If $s_3>0$, the third equation of \eqref{eq:sum_u_v1} implies $v^*_3\le k$. Then, 
$$
(d+q)v^*_1-dv^*_3\ge(d+q)k-dk>0,
$$
which is a contradiction. If $s_3=0$, then $(d+q)v^*_1-dv^*_3=0$. Again by $v_1^*\ge k$ and the first equation of \eqref{eq:sum_u_v1}, $(d+q)v^*_1-dv^*_2\le 0$. This leads to contradiction by the second equation of \eqref{eq:sum_u_v1} and \eqref{v2ss}.
\end{proof}

\begin{lemma}\label{lemma_unstable}
    Suppose that $D$ and $Q$ are given by \eqref{DQ}. Let $\bm r=(r, 0, 0)$ and $\bm s=(s_1, s_2, s_3)>\bm 0$ with $\bm s\neq \bm r$ and $\sum_{i=1}^3 s_i=r$. Then the semitrivial equilibrium $E_2=(\bm 0, \bm v^*)$ is unstable and the semitrivial equilibrium $E_1=(\bm u^*, \bm 0)$ is stable for model \eqref{patch}. 
\end{lemma}
\begin{proof}
    The stability of $E_2$ is determined by the sign of  $\la_1(\bm{r}, \bm v^*)$, which is the principal eigenvalue of  $dD+qQ+\text{diag}(r_i(1-v^*_i/k))$.
    By the assumptions on $\bm r$  and Lemma \ref{lemma_sign1}, we have 
$$
r_1\left(1-\frac{v^*_1}{k}\right)>0
$$
and 
$$
r_i\left(1-\frac{v^*_i}{k}\right)=0, \ i=2, 3.
$$
Therefore, by Lemma \ref{lemma_sp}, we must have $\la_1(\bm{r}, \bm v^*)>s(dD+qQ)=0$. Hence, $E_2$ is  unstable.

 The stability of $E_1$ is determined by the sign of  $\la_1(\bm{s}, \bm u^*)$, which is the principal eigenvalue of  $dD+qQ+\text{diag}(s_i(1-u^*_i/k))$. By \cite{nguyen2022population}, we have $\bm u^*=(k, (d+q)k/d, (d+q)k/d)$. So, 
$$
s_1\left(1-\frac{u^*_1}{k}\right)=0,
$$
and 
$$
s_i\left(1-\frac{u^*_i}{k}\right)\le 0, \ i=2, 3,
$$
with at least one strict sign by the assumption on $\bm s$. Therefore, by Lemma \ref{lemma_sp}, we must have $\la_1(\bm{s}, \bm u^*)<s(dD+qQ)=0$. Hence, $E_1$ is  stable. 
\end{proof}

By the theory of monotone dynamical systems (Lemma \ref{M1}) and Lemmas \ref{lemma_nonexistence} and \ref{lemma_unstable}, we obtain the following result:
\begin{theorem}
    Suppose that $D$ and $Q$ are given by \eqref{DQ}. Let $\bm r=(r, 0, 0)$ and $\bm s=(s_1, s_2, s_3)>\bm 0$ with $\bm s\neq \bm r$ and $\sum_{i=1}^3 s_i=r>0$. Then the semitrivial equilibrium $E_1=(\bm u^*, \bm 0)$ is globally asymptotically stable for model \eqref{patch}. 
\end{theorem}

\section{Stream networks of $n$ nodes}

In this section, we generalize the results in Section \ref{sec3nodes} to a certain type of network of $n$ nodes. We recall the definition of stream networks of  $n$ nodes in \cite{nguyen2022population}.

\begin{definition}\label{def:leveled_graph}
Let $G$ be a directed graph, and denote the set of nodes of $G$ by $V$. Consider a function $f: V\to \mathbb{Z}_{\geq 0}$. For each node $i$, we call $f(i)$ the \textit{level} of the node and $(G,f)$  a \textit{leveled graph} if the following assumptions are satisfied
\begin{enumerate}
    \item[(i)] For each $0\leq k\leq \max_{i\in V}\{f(i)\}$, there exists a node $j$ such that $f(j)=k$.
    \item[(ii)] For each pair of nodes $i$ and $j$, there is no edge between $i$ and $j$ if $|f(i)-f(j)|\neq 1$.
\end{enumerate}
\end{definition}

We use level graphs to describe a type of stream network, where the nodes in further downstream positions have larger levels. The left digraph in Figure \ref{leveled_graph} is a leveled graph while the right digraph is not. 

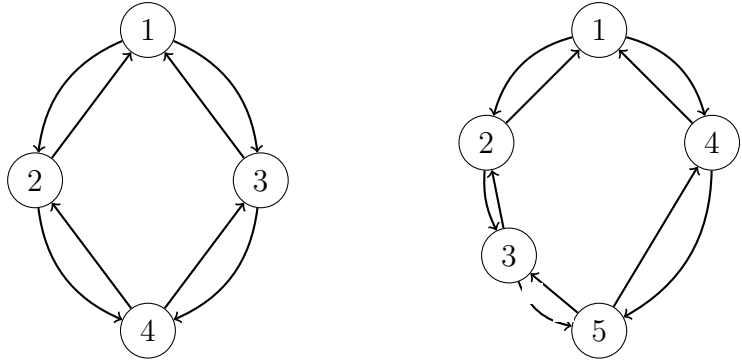
\begin{figure}[htbp]
\centering
\begin{tikzpicture}
\begin{scope}[every node/.style={draw}, node distance= 1.5 cm]

    \node[circle] (1) at (-3+0,     0-0) {$1$};
    \node[circle] (2) at (-3-1.5,   0-2) {$2$};
    \node[circle] (3) at (-3+1.5,   0-2) {$3$};
    \node[circle] (4) at (-3+0,     0-4) {$4$};

    \node[circle] (5) at (3+0,      0-0) {$1$};
    \node[circle] (6) at (3-1.5,    0-1.5) {$2$};
    \node[circle] (7) at (3-1.2,    0-3) {$3$};
    \node[circle] (8) at (3+1.5,    0-1.5) {$4$};
    \node[circle] (9) at (3+0,      0-4) {$5$};
 
\end{scope}
\begin{scope}[every node/.style={fill=white},
              every edge/.style={thick}]
    
    \draw[thick] [->](1) to [bend right] node[left=5] {{}} (2); 
    \draw[thick] [<-](1) to node[right=4] {{}} (2); 
    \draw[thick]  [->](1) to [bend left] node[right=5] {{}} (3); 
    \draw[thick] [<-](1) to node[left=4] {{}} (3); 
    \draw[thick] [->](2) to [bend right] node[right=5] {{}} (4); 
    \draw[thick] [<-](2) to node[left=4] {{}} (4); 
    \draw[thick] [->](3) to [bend left] node[right=5] {{}} (4); 
    \draw[thick] [<-](3) to node[left=4] {{}} (4); 

    \draw[thick] [->](5) to [bend right] node[left=5] {{}} (6); 
    \draw[thick] [<-](5) to node[right=4] {{}} (6); 
    \draw[thick] [->](6) to [bend right=15] node[right=5] {{}} (7); 
    \draw[thick] [<-](6) to (7); 
    \draw[thick] [->](5) to [bend left] node[right=5] {{}} (8); 
    \draw[thick] [<-](5) to node[left=4] {{}} (8); 
    \draw[thick] [->](8) to [bend left] node[right=5] {{}} (9); 
    \draw[thick] [<-](8) to node[left=4] {{}} (9); 
    \draw[thick] [->](7) to [bend right] node[right=5] {{}} (9); 
    \draw[thick] [<-](7) to node[left=4] {{}} (9); 

\end{scope}
\end{tikzpicture}
\caption{The left digraph is a leveled graph with level function $f(1)=0,f(2)=f(3)=1,f(4)=2$. The right digraph cannot be a leveled graph for any choice of level function.}\label{leveled_graph}

\end{figure}

\begin{definition}[\cite{nguyen2022population}]
Consider a graph $G$ with level function $f$ and connection matrix $L$. We say that $(G,f, L)$ is a \textit{homogeneous flow stream network} if the following assumptions are satisfied:
\begin{enumerate}
    \item[(i)] The matrix $L$ is irreducible.
    \item[(ii)] If there is an edge from node $i$ to node $j$, then there is also an edge from node $j$ to node $i$. 
    \item[(iii)] If there is an edge from node $i$ to node $j$, then the weight is $\ell_{ij} = d+q$ if $f(j)-f(i)=1$ (i.e. the edge is from an upstream to a downstream node) and $\ell_{ij}=d$ if $f(i)-f(j)=1$ (i.e. the edge is from a downstream to an upstream node). Here, $d$ and $q$ are positive constants.
\end{enumerate}
\end{definition}

The connection matrix $L$ of a homogeneous flow stream network can be written as $L=dD+qQ$. We recall the following result about the positive eigenvector of $L$ proved in \cite{nguyen2022population}.
\begin{lemma}\label{lemma vector v}
Let $(G,f,L)$ be a homogeneous flow stream network. Let $\bm v$ be the solution to $L\bm v=(dD+qQ)\bm v = 0$. Then the eigenvector $\bm v$ writes, up to a constant multiple, as
\[
v_i = \left(\frac{d+q}{d}\right)^{f(i)}.
\]
\end{lemma}

Let $\bm u^*=(u_1^*, \dots, u_n^*)$ be the positive equilibrium of 
\begin{equation}\label{pat-sn}
\ds\frac{du_i}{dt}=\ds\sum_{j=1}^{n}(dD_{ij}+qQ_{ij})u_j+r_iu_i\left(1-\frac{u_i}{k}\right), \ \ i=1,\dots, n,\;\;t>0.
\end{equation}
The total biomass $\mathcal{K}$ of $\bm u^*$ is defined as $\mathcal{K}:=\sum_{i=1}^n u_i^*$. We also recall the following  theorem in \cite{nguyen2022population} about the total biomass $\mathcal{K}$:

\begin{theorem}\label{theorem:biomass_n}
Let $(G,f,L)$ be a homogeneous flow stream network, where $L=dD+qQ$. Suppose that $\bm r=(r_1,\dots, r_n)>\bm 0$ with $\sum_{i=1}^nr_i=r>0$ and $k>0$. Then the total biomass $\mathcal{K}$ of the positive equilibrium of \eqref{pat-sn} has the upper bound
\[
\mathcal{K} \leq  k\sum_{i=1}^n \left(1+\frac{q}{d}\right)^{f(i)}.
\]
Moreover, the maximum is achieved as the upper bound when $r_i=0$ for any node $i$ with positive level, i.e. $f(i)>0$.  In this case,  $u_i^*=k(\frac{d+q}{d})^{f(i)}$.
\end{theorem}

The main result we prove in this section is the following, which states that to gain a competitive advantage in a homogeneous flow stream network one needs to distribute all the resources to the upstream ends, i.e. nodes with level 0. 

\begin{theorem}\label{theorem_main}
   Let $(G,f,L)$ be a homogeneous flow stream network, where $L=dD+qQ$. Let $k>0$ and $\bm r, \bm s>\bm 0$ such that $\sum_{i=1}^n r_i=\sum_{i=1}^n s_i=r>0$.  
   Suppose that $r_i=0$ for any node $i$ with $f(i)>0$ and there exists at least one node $i_0$ with $f(i_0)>0$ such that $s_{i_0}>0$.
   Then the semitrivial equilibrium $E_1=(\bm u^*, \bm 0)$ is globally asymptotically stable for \eqref{patch}. 
\end{theorem}


\subsection{Proof of Theorem \ref{theorem_main}}
Suppose to the contrary that $(\bm u,\bm v)$ is a positive equilibrium of \eqref{patch}. Let  $w_i^*=u_i+v_i$ for each $i=1, \dots, n$.
To show the non-existence of a positive equilibrium $E^*$, we recall the sign pattern approach used in \cite{nguyen2022population}. We associate the stream network with a \textit{sign pattern graph}. The nodes of the sign pattern graph are the same as the nodes in the stream network, but additionally we assign each node $i$ in the sign pattern graph  with a sign based on the value of $w_i^*$:
\[
\text{sign}(\text{node } i) = \begin{cases}
    + \quad &\text{if} \quad w_i^* <k \quad \text{and} \quad r_i>0 \text{ or } s_i>0\\
    - \quad &\text{if} \quad w_i^* >k \quad \text{and} \quad r_i>0 \text{ or } s_i>0\\
    0 \quad &\text{if} \quad w_i^*=k\\
    0^+ \quad &\text{if} \quad w_i^*<k \quad \text{and} \quad r_i=s_i=0\\
    0^- \quad &\text{if} \quad w_i^*>k \quad \text{and} \quad r_i=s_i=0.\\
\end{cases}
\]
Next, if there is an edge between node $i$ and an adjacent, downstream node $j$ in the stream network, we draw an edge between node $i$ and node $j$ in the sign pattern graph as follows:
\begin{enumerate}
    \item There is a directed edge from node $i$ to node $j$ (adjacent, downstream of node $i$), denoted $i\to j$, if 
    \[
    (d+q)w_i^* > dw_j^*.
    \]
    \item There is a directed edge from node $j$ to node $i$ (adjacent, downstream of node $i$), denoted $j\to i$, if
    \[
    (d+q)w_i^* < dw_j^*.
    \]
    \item There is an undirected edge between node $i$ and node $j$, denoted $i - j$, if
    \[
    (d+q)w_i^* = dw_j^*.
    \]
\end{enumerate} 
The edges in the sign pattern graph describe the \textit{net flow} between adjacent nodes in the stream network.
\begin{lemma}\label{lemma:edge}
For any $(+)$ node, there must be at least one directed edge out of the node. For any $(-)$ node, there must be at least one directed edge into the node. For any $(0),(0^+),(0^-)$ node, either all edges connected to the node are undirected, or there must be at least one edge into and one edge out of the node.
\end{lemma}
\begin{proof}
Suppose node $i$ has sign $(+)$. We add the equations of $u_i$ and $v_i$  to obtain
\[
0=\sum_{j: f(j)-f(i)=1}(dw_j^* - (d+q)w_i^*) + \sum_{j: f(j)-f(i)=-1}((d+q)w_j^* - dw_i^*) + (r_iu_i+s_iv_i)\bigg(1-\frac{w_i^*}{k}\bigg).
\]
Since $w_i^* < k$, there must be a negative term in the sum above corresponding to a node $j$ adjacent to node $i$. It is easy to check that whether $j$ is upstream or downstream of $i$ we always have an edge $i\to j$ in the sign pattern graph. We can repeat the same argument for the nodes with sign $(-), (0), (0^+), (0^-)$.
\end{proof} 

\begin{lemma}\label{lemma:inequality}
If $i-j$ or $i\to j$, then
\[
w_i^* \geq \bigg(\frac{d+q}{d}\bigg)^{f(i)-f(j)}w_j^*. 
\]
\end{lemma}

\begin{proof}
If node $i$ is upstream of node $j$, we have $f(i)-f(j)=-1$. Since $i-j$ or $i\to j$, by the way we assign edges to the sign pattern graph we must have
\[
w_i^* \geq \frac{d}{d+q} w_j^* = \bigg(\frac{d+q}{d}\bigg)^{f(i)-f(j)}w_j^*. 
\]
If node $i$ is downstream of node $j$, we have $f(i)-f(j)=1$. Since $i-j$ or $i\to j$, again we have
\[
w_i^* \geq \frac{d+q}{d} w_j^* = \bigg(\frac{d+q}{d}\bigg)^{f(i)-f(j)}w_j^*. 
\]
\end{proof}
The following corollary follows directly from Lemma \ref{lemma:inequality}.
\begin{corollary}\label{cor:updown}
If node $i$ is downstream of node $j$ (i.e. $f(i)-f(j)=1$) and $i-j$ or $i\to j$ then $w_i^* > w_j^*$. 

\end{corollary}

\begin{corollary}\label{cor:path}
    If there is a path from node $i$ to node $j$, i.e. there exist nodes $i_1,i_2,\dots,i_h$ such that $i - (or\to) i_1 - (or\to) \dots - (or\to) i_h - (or\to) j$, then
    \[
    w_i^* \geq \bigg(\frac{d+q}{d}\bigg)^{f(i)-f(j)}w_j^*. 
    \]
    The equality happens when all edges are undirectedly  connected.
\end{corollary}
\begin{proof}
    We apply Lemma \ref{lemma:inequality} repeatedly
    \begin{align*}
    w_i^* &\geq \bigg(\frac{d+q}{d}\bigg)^{f(i)-f(i_1)}w_{i_1}^*\\
    &\geq \bigg(\frac{d+q}{d}\bigg)^{f(i)-f(i_1)}\bigg(\frac{d+q}{d}\bigg)^{f(i_1)-f(i_2)}w_{i_2}^*\\
    &\vdots \\
    &\geq \bigg(\frac{d+q}{d}\bigg)^{f(i)-f(i_1)+f(i_1)-f(i_2)+\dots+f(i_h)-f(j)}w_j^*\\
    &= \bigg(\frac{d+q}{d}\bigg)^{f(i)-f(j)}w_j^*. 
    \end{align*}
    Since each equality happens when the corresponding edge is $-$, the overall equality happens when all the edges in the path are $-$.
\end{proof}
\begin{corollary}\label{cor:cycle}
If there is a cycle $i - (or\to) i_1 - (or\to) \dots - (or\to) i_h - (or\to) i$, then all the edges in the cycle must be $-$.
\end{corollary}
\begin{proof}
    The proof follows directly from Corollary \ref{cor:path} where we set $j=i$.
\end{proof}

\begin{lemma}\label{lemma:mostupstream}
    For any $i$ such that $f(i)=0$  (i.e. the most upstream nodes) and $r_i>0$, we have $w_i^* \leq k$. 
\end{lemma}
\begin{proof}
Assume by contradiction that there exists a node $i$ such that $f(i)=0$, $r_i>0$ and $w_i^* >k$. Then node $i$ has  $(-)$  sign and by Lemma \ref{lemma:edge}, there must exist a node $i_1$ such that  $i_1\to i$. Since node $i$ is a most upstream node, node $i_1$ must be downstream of it and thus by Corollary \ref{cor:updown}, we must have $w_{i_1}^* > w_i^* >k$, thus node $i_1$ has either sign $(-)$ or $(0^-)$.

Since there is already an edge out of node $i_1$, this means there exists a node $i_2$ such that  $i_2 \to i_1$. This gives us a path from $i_2$ to $i$ and thus by Corollary \ref{cor:path} we must have
\[
w_{i_2}^* \geq \bigg(\frac{d+q}{d}\bigg)^{f(i_2)-f(i)}w_i^* = \bigg(\frac{d+q}{d}\bigg)^{f(i_2)}w_i^* >k. 
\]
since $f(i)=0$. Thus again we must have $i_2$ has either sign $(-)$ or $(0^-)$. 

For each index $h\geq 3$, we can repeat the argument to obtain node $i_h$ such that $i_h \to i_{h-1}$ and node $i_h$ has either sign $(-)$ or $(0^-)$. Since the number of nodes is finite, the above process must stop after a finite number of steps. This is only possible if we have a cycle. However, since all edges in this cycle are directed, we reach  a contradiction based on Corollary \ref{cor:cycle}.
\end{proof}

\begin{lemma}\label{lemma_no}
Let $(G,f,L)$ be a homogeneous flow stream network, where $L=dD+qQ$. Let $k>0$ and $\bm r, \bm s>\bm 0$ such that $\sum_{i=1}^n r_i=\sum_{i=1}^n s_i=r>0$.  
Suppose that $r_i=0$ for any node $i$ with $f(i)>0$ and there exists at least one node $i_0$ with $f(i_0)>0$ such that $s_{i_0}>0$. Then a positive equilibrium $E^*$ does not exist.
\end{lemma}
\begin{proof}
Suppose by contradiction that a positive equilibrium $E^*$ exists. From the assumption on $\bm r$, there must exists a node $i$ with $f(i)=0$ and $r_i>0$. Taking the sum of all $du_i/dt$,  at the positive equilibrium, we must have
\begin{align}\label{eq:sum_u}
0 = \sum_{i:f(i)=0, r_i>0} r_i u_i\bigg(1-\frac{w_i^*}{k}\bigg).
\end{align}
From Lemma \ref{lemma:mostupstream}, for any $i$ such that $f(i)=0$ and $r_i>0$ we must have $w_i^*\leq k$. From this fact and equation \eqref{eq:sum_u}, we have $w_i^*=k$ for any $i$ such that $f(i)=0$ and $r_i>0$. 

 Without loss of generality, suppose that $f(1)=0$ and $r_1>0$. From the argument above, $w_1^*=k$ and node $1$ has sign $(0)$. If there is a node $j$ adjacent to node $1$ such that $j\to 1$, then repeating the argument in Lemma \ref{lemma:mostupstream}, we have a cycle which leads to a contradiction. Since there is no directed edge into node $1$, from Lemma \ref{lemma:edge}, all edges connected to node $1$ must be undirected. Let $j_1,\dots,j_h$ denote the nodes adjacent to node $1$.  Then from Corollary \ref{cor:updown}, they must have sign $(-)$ or $(0^-)$.

We will show again that all edges connected to node $j_1$ must be undirected, and a similar argument can be applied to show all edges connected to nodes $j_1,\dots,j_h$ are undirected. Assume by contradiction that not all edges connected to node $j_1$ are undirected. Since node $j_1$ has sign $(-)$ or $(0^-)$, from Lemma \ref{lemma:edge} there must be a directed edge into node $j_1$ from another node $j'$. However, that means there is a path $j' \to j_1 - 1$. From Corollary \ref{cor:path} we have
\[
w_{j'}^* >\bigg(\frac{d+q}{d}\bigg)^{f(j')-f(1)} w^*_1 = \bigg(\frac{d+q}{d}\bigg)^{f(j')} w^*_1\geq w^*_1 = k,
\]
and thus node $j'$ has sign $(-)$ or $(0^-)$ and there is a directed edge into it. We repeat the argument in Lemma 4.11, which leads to a cycle and thus a contradiction.

The same argument can be repeated, and since the stream network is strongly connected, we have all edges in the sign pattern graph must be undirected. Thus all nodes aside from the most upstream nodes must have sign $(-)$ or $(0^-)$. Since there exists $i_0$ such that $f(i_0)>0$ and $s_{i_0}>0$, there must be at least one node with sign $(-)$. Taking the sum of all equations for $u_i$ and $v_i$, we have
\[
0=\sum_{i=1}^n (r_iu_i+s_iv_i)\bigg(1-\frac{w_i^*}{k}\bigg).
\]
However, since there is at least one node with sign $(-)$ and no node with sign $(+)$, the right hand side of the equation above must be strictly negative, which is a contradiction. Thus a positive equilibrium $E^*$ does not exists.
\end{proof}

The proof of the following result is similar to that of Theorem \ref{theorem:biomass_n}. We include it here for the sake of completeness. 
\begin{lemma}\label{lemma_len}
   Let $(G,f,L)$ be a homogeneous flow stream network, where $L=dD+qQ$. Suppose $\bm k=(k, \dots, k)$ with $k>0$ and $\bm s>\bm 0$.   Let $\bm v^*=(v^*_1, \dots,  v^*_n)$ be the positive equilibrium of \eqref{DS-single}.
   If there exists at least one node $i_0$ with  $f(i_0)>0$ such that $s_{i_0}>0$, 
   then $v^*_i<k$ for all node $i$ with $f(i)=0$. 
\end{lemma}
\begin{proof}
Since $L=dD+qQ$ is essentially nonnegative and irreducible, by  \cite[Theorem 4.1.1]{smith2008monotone}, the solutions of \eqref{patch} induce a strongly monotone dynamical system: if $\bm u_1(0)>\bm u_2(0)$ then the corresponding solutions satisfy  $\bm u_1(t)\gg \bm u_2(t)$ for all $t>0$. By Lemma \ref{lemma vector v}, $dD+qQ$ has a positive eigenvector $\bm v=(v_1, \dots, v_n)$ such that $v_i=1$ if $f(i)=0$ for all $i=1, \dots, n$. Moreover, $v_i>1$ if $f(i)>0$. 
Define $\bar{\bm u}=k\bm v$. Since $\bar {\bm u}$ is an eigenvector of $dD+qQ$ corresponding to eigenvalue 0,  we have
$$
0\ge s_i\bar u_i\left(1-\frac{\bar u_i}{k}\right)= \sum_{j=1}^n (dD_{ij}+qQ_{ij})\bar u_j+s_i\bar u_i\left(1-\frac{\bar u_i}{k}\right), \ \ i=1,\dots, n.
$$
Moreover, the $i_0$-th inequality is strict since $s_{i_0}>0$ and $v_{i_0}>1$. Hence, the solution $\bm u(t)$ of \eqref{patch} with initial condition $\bm u(0)=\bm {\bar u}$ is  strictly decreasing and converges to  an equilibrium \cite[Proposition 3.2.1]{smith2008monotone}, which is the positive equilibrium $\bm v^*$ by Lemma \ref{DS-single}. Hence, $\bm v^*\ll \bar {\bm u}$. In particular, $v_i^*<k$ if $f(i)=0$.
\end{proof}

We are now ready to prove that $E_2$ is unstable. 
\begin{lemma}\label{lemma_les}
   Let $(G,f,L)$ be a homogeneous flow stream network, where $L=dD+qQ$. Let $k>0$ and $\bm r, \bm s>\bm 0$ such that $\sum_{i=1}^n r_i=\sum_{i=1}^n s_i=r>0$.  
   Suppose that $r_i=0$ for any node $i$ with $f(i)>0$ and there exists at least one node $i_0$ with  $f(i_0)>0$ such that $s_{i_0}>0$.
   Then the semitrivial equilibrium $E_2=(\bm 0, \bm v^*)$ of \eqref{patch} is unstable  and the semitrivial equilibrium $E_1=(\bm u^*, \bm 0)$ of \eqref{patch} is stable. 
\end{lemma}
\begin{proof}
To see that $E_2$ is unstable,   it suffices to show $\la_1:=\la_1(\bm{r}, \bm v^*)>0$, where $\la_1(\bm{r}, \bm v^*)$ is the principal eigenvalue of 
the matrix $dD+qQ+\text{diag}(r_i(1-v^*_i/k))$. Let $\bm\varphi=(\varphi_1, \dots, \varphi_n)$ be a positive eigenvector corresponding with $\lambda_1$. Then,
   $$
   \lambda_1 \varphi_i=\sum_{j=1}^n (dD_{ij}+qQ_{ij})\varphi_j+r_i\left(1-\frac{v_i^*}{k}\right)\varphi_i.
   $$ 
   Adding up all the equations and noticing that each column sum of $dD+qQ$ is zero, we obtain 
   \begin{eqnarray*}
  \lambda_1 \sum_{i=1}^n \varphi_i&=& \sum_{i=1}^n r_i\left(1-\frac{v_i^*}{k}\right)\varphi_i \\
  &=& \sum_{i:\ f(i)=0} r_i\left(1-\frac{v_i^*}{k}\right)\varphi_i,
   \end{eqnarray*}
   where we used the assumption that $r_i=0$ if $f(i)>0$ in the last step. By Lemma \ref{lemma_len}, $v_i^*<k$ if $f(i)=0$. Therefore, we have $\lambda_1>0$.

 The stability of $E_1$ is determined by the sign of  $\la_1(\bm{s}, \bm u^*)$, which is the principal eigenvalue of  $dD+qQ+\text{diag}(s_i(1-u^*_i/k))$. By Theorem \ref{theorem:biomass_n}, we have $ u_i^*=k(\frac{d+q}{d})^{f(i)}$. So we have
$$
s_i\left(1-\frac{u^*_i}{k}\right)=0, \ \forall i \text{ such that } f(i)=0
$$
and 
$$
s_i\left(1-\frac{u^*_i}{k}\right)\le 0, \ \forall i \text{ such that } f(i)>0,
$$
with at least one strict sign due to the assumption on $\bm s$. Therefore, by Lemma \ref{lemma_sp}, we must have $\la_1(\bm{s}, \bm u^*)<s(dD+qQ)=0$. Hence, $E_1$ is  stable. 
\end{proof}

Finally, Theorem \ref{theorem_main} follows from Lemmas \ref{M1}, \ref{lemma_no}, and \ref{lemma_les}.

 \section{Appendix}

\subsection{Results on configuration (ii)}

For configuration (ii) we show that if $\bm r=(r_1, r_2, 0)$ with $r_1+r_2=r$ and $\bm s=(s_1, s_2, s_3)$ with $\sum_{i=1}^3 s_i=r$ and $s_3\neq 0$ then $E_1$ is globally asymptotically stable.

\begin{lemma}\label{lemma_signa}
    Suppose that $D$ and $Q$ are given by \eqref{DQ1}. Let $\bm r=(r_1, r_2, 0)>\bm 0$ and $\bm s=(s_1, s_2, s_3)>\bm 0$ such that $\sum_{i=1}^2 r_i=\sum_{i=1}^3 s_i=r$ and $s_3>0$. If $(\bm u, \bm v)$ is a positive equilibrium of \eqref{patch}, then $u_1+v_1<k$, $u_2+v_2<k$, and $u_3+v_3>k$.
\end{lemma}
\begin{proof}
Suppose that $(\bm u, \bm v)$ is a positive equilibrium of \eqref{patch}. Let $w_i:=u_i+v_i$ for $i=1, 2, 3$. Then, we have
\begin{eqnarray}
&&-(d+q)w_1+dw_3+(r_1u_1+s_1v_1)\left(1-\frac{w_1}{k} \right)=0, \nonumber\\
&&-(d+q)w_2+dw_3+(r_2u_2+s_2v_2)\left(1-\frac{w_2}{k} \right)=0, \label{eq:2}\\
&&(d+q)w_1-dw_3+(d+q)w_2-dw_3+(r_3u_3+s_3v_3)\left(1-\frac{w_3}{k} \right)=0. \nonumber
\end{eqnarray}
Assume to the contrary that $w_3\le k$. Then by the third equation of \eqref{eq:2}, we have either $(d+q)w_1-dw_3\le 0$ or  $(d+q)w_2-dw_3\le 0$. Without loss of generality, we may assume  $(d+q)w_1-dw_3\le 0$. This implies that $w_1\le dw_3/(d+q)<k$. If $r_1>0$ or $s_1>0$, then 
$$
-(d+q)w_1+dw_3+(r_1u_1+s_1v_1)\left(1-\frac{w_1}{k} \right)>0
$$
which contradicts the first equation of \eqref{eq:2}. If $r_1=s_1=0$, then $r_2>0$ and $(d+q)w_1-dw_3= 0$ by the first equation of \eqref{eq:2}. Then by $w_3\le k$  and the third equation of \eqref{eq:2} again, we have $(d+q)w_2-dw_3\le 0$. By $r_2>0$ and the second equation of \eqref{eq:2}, we have $w_2\ge k$. Therefore,
$$
(d+q)w_2-dw_3\ge (d+q)k-dk>0,
$$
which is a contradiction. Hence, $w_3=u_3+v_3>k$. 

By $w_3>k$, either $(d+q)w_1-dw_3>0$ or $(d+q)w_2-dw_3>0$. Without loss of generality, say $(d+q)w_1-dw_3>0$. Then by the first equation of \eqref{eq:2}, we have $w_1=u_1+v_1<k$. Suppose to the contrary that $w_2\ge k$. Then by the second equation of \eqref{eq:2}, $(d+q)w_2-dw_3\le 0$. Since $w_3<(d+q)w_1/d<(d+q)k/d$, 
$$
0\ge (d+q)w_2-dw_3> (d+q)k-d \frac{(d+q)k}{d}=0,
$$
which is a contradiction. Therefore, $w_2=u_2+v_2<k$.
\end{proof}

Then we show the non-existence of a positive equilibrium. 
\begin{lemma}\label{lemma_nonexistence1}
    Suppose that $D$ and $Q$ are given by \eqref{DQ1}. Let $\bm r=(r_1, r_2, 0)>\bm 0$ and $\bm s=(s_1, s_2, s_3)>\bm 0$ such that $\sum_{i=1}^2 r_i=\sum_{i=1}^3 s_i=r>0$ and $s_3>0$. Then model \eqref{patch} has no positive equilibrium. 
\end{lemma}
\begin{proof}
    Suppose to the contrary that $(\bm u, \bm v)$ is a positive equilibrium of \eqref{patch}. Since $r_3=0$,  $(\bm u, \bm v)$ satisfies 
    \begin{equation}\label{eql1}
\begin{cases}
\ds -(d+q)u_1+du_3+r_1u_1\left(1-\frac{u_1+v_1}{k}\right)=0, \\
\ds -(d+q)u_2+du_3+r_2u_2\left(1-\frac{u_2+v_2}{k}\right)=0, \\
((d+q)u_1-du_3)+((d+q)u_2-du_3)=0.
\end{cases}
\end{equation}
Adding up the equations in \eqref{eql1}, we obtain
\begin{equation}\label{sum1}
\sum_{i=1}^2 r_iu_i\left(1 -\frac{u_i+v_i}{k}\right)=0.
\end{equation}
By Lemma  \ref{lemma_signa} and $r_1+r_2>0$, the left hand side of \eqref{sum1} is positive, which is a contradiction. 
\end{proof}

In the following two lemmas, we show that the semitrivial equilibrium $E_2$ is unstable. 

\begin{lemma}\label{lemma_signb}
    Suppose that $D$ and $Q$ are given by \eqref{DQ1}. Let $\bm s=(s_1, s_2, s_3)\ge\bm 0$ with $s_3>0$. Then the semitrivial equilibrium $E_2=(\bm 0, \bm v^*)$ satisfies $v^*_1<k$ and $v^*_2<k$.
\end{lemma}

\begin{proof}
We observe that $\bm v^*$ must satisfy
\[
0= \sum_{j=1}^3(dD_{ij}+qQ_{ij})v^*_j + s_iv_i^*\left(1-\frac{v_i^*}{k}\right), \ \ i=1, 2, 3.
\]
That is
\begin{eqnarray}
&&0=-(d+q)v^*_1+dv^*_3+s_1v_1^*\left(1-\frac{v^*_1}{k} \right), \nonumber\\
&&0=-(d+q)v^*_2+dv^*_3+s_2v_2^*\left(1-\frac{v_2^*}{k} \right), \label{eq:22}\\
&&0=(d+q)v_1^*-dv^*_3+(d+q)v^*_2-dv^*_3+s_3v_3^*\left(1-\frac{v_3^*}{k} \right). \nonumber
\end{eqnarray}
Assume to the contrary that $v_1^*\ge k$. Then by the first equation of \eqref{eq:22}, $(d+q)v^*_1-dv^*_3\le 0$. This implies $v_3^*\ge(d+q)k/d>k$. By $s_3>0$ and the third equation of \eqref{eq:22}, either $(d+q)v_1^*-dv^*_3>0$ or $(d+q)v_2^*-dv^*_3>0$. Hence, $(d+q)v_2^*-dv^*_3>0$. Then by the second equation of \eqref{eq:22},  we have $v_2^*<k$ and 
$$
0<(d+q)v_2^*-dv_3^*<(d+q)k-d\frac{(d+q)k}{d}=0,
$$
which is a contradiction. Therefore, $v_1^*<k$. Similarly, $v_2^*<k$.
\end{proof}

\begin{lemma}\label{lemma_unstable1}
    Suppose that $D$ and $Q$ are given by \eqref{DQ1}. Let $\bm r=(r_1, r_2, 0)>\bm 0$ and $\bm s=(s_1, s_2, s_3)>\bm 0$ such that $\sum_{i=1}^2 r_i=\sum_{i=1}^3 s_i=r>0$ and $s_3>0$. Then the semitrivial equilibrium $E_2=(\bm 0, \bm v^*)$ is unstable and  the semitrivial equilibrium $E_1=(\bm u^*, \bm 0)$ is stable for model \eqref{patch}. 
\end{lemma}
\begin{proof}
    The stability of $E_2$ is determined by the sign of  $\la_1(\bm{r}, \bm v^*)$, which is the principal eigenvalue of  matrix $dD+qQ+\text{diag}(r_i(1-v^*_i/k))$.  
    By the assumptions on $\bm r$ and Lemma \ref{lemma_signb}, we have 
$$
r_i\left(1-\frac{v^*_i}{k}\right)\ge 0, \ \ i=1, 2,
$$
with at least one strict inequality
and 
$$
r_3\left(1-\frac{v^*_3}{k}\right)=0.
$$
Therefore, by Lemma \ref{lemma_sp}, we must have $\la_1(\bm{r}, \bm v^*)>s(dD+qQ)=0$. Hence, $E_2$ is  unstable.

 The stability of $E_1$ is determined by the sign of  $\la_1(\bm{s}, \bm u^*)$, which is the principal eigenvalue of  $dD+qQ+\text{diag}(s_i(1-u^*_i/k))$. By \cite{nguyen2022population}, we have $\bm u^*=(k, k, (d+q)k/d)$. So, 
$$
s_i\left(1-\frac{u^*_i}{k}\right)=0, \ i=1, 2,
$$
and 
$$
s_3\left(1-\frac{u^*_3}{k}\right)< 0, 
$$
since $s_3>0$. Therefore, by Lemma \ref{lemma_sp}, we must have $\la_1(\bm{s}, \bm u^*)<s(dD+qQ)=0$. Hence, $E_1$ is  stable. 
\end{proof}

By the theory of monotone dynamical systems (see Lemma \ref{M1}) and Lemmas \ref{lemma_nonexistence1} and \ref{lemma_unstable1}, we obtain the following result:
\begin{theorem}
    Suppose that $D$ and $Q$ are given by \eqref{DQ1}. Let $\bm r=(r_1, r_2, 0)>\bm 0$ and $\bm s=(s_1, s_2, s_3)>\bm 0$ such that $\sum_{i=1}^2 r_i=\sum_{i=1}^3 s_i=r>0$ and $s_3>0$. Then the semitrivial equilibrium $E_1=(\bm u^*, \bm 0)$ is globally asymptotically stable for model \eqref{patch}. 
\end{theorem}

\subsection{Results on configuration (iii)}

Finally, for configuration (iii) we show that if $\bm r=(r, 0, 0)$ and $\bm s\neq\bm r$, $E_1$ is always globally asymptotically stable.

In the following two lemmas, we show that the model has no positive equilibrium. 
\begin{lemma}\label{lemma_signc}
    Suppose that $D$ and $Q$ are given by \eqref{DQ2}. Let $\bm r=(r, 0, 0)$ and $\bm s=(s_1, s_2, s_3)>\bm 0$ with $\bm s\neq \bm r$ and $\sum_{i=1}^3 s_i=r>0$. 
    If $(\bm u, \bm v)$ is a positive equilibrium of \eqref{patch}, then $u_1+v_1<k$  and $u_3+v_3>k$.
\end{lemma}
\begin{proof}
Suppose that $(\bm u, \bm v)$ is a positive equilibrium of \eqref{patch}.  Let $w_i=u_i+v_i$ for $i=1, 2, 3$. Then, we have
\begin{eqnarray}
&&-(d+q)w_1+dw_2+(r_1u_1+s_1v_1)\left(1-\frac{w_1}{k} \right)=0, \nonumber\\
&&(d+q)w_1-dw_2-(d+q)w_2+dw_3+(r_2u_2+s_2v_2)\left(1-\frac{w_2}{k} \right)=0, \\
&&(d+q)w_2-dw_3+(r_3u_3+s_3v_3)\left(1-\frac{w_3}{k} \right)=0. \nonumber
\end{eqnarray}\label{eq:3}
Suppose to the contrary that $w_1\ge k$. Then by the first equation of \eqref{eq:3}, $(d+q)w_1-dw_2\le 0$. So $w_2\ge (d+q)k/d>k$. So by the second equation \eqref{eq:3}, $(d+q)w_2-dw_3\le 0$ and the inequality is strict if $s_2>0$. Hence, $w_3\ge (d+q)k/d>k$. If $s_3>0$, then $(r_3u_3+s_3v_3)(1-{w_3}/{k} )<0$ and the third equation of \eqref{eq:3} leads to a contradiction. If $s_3=0$, by the assumptions on $\bm s$, $s_2\neq 0$ and $(d+q)w_2-dw_3< 0$. Then  $(r_3u_3+s_3v_3)(1-{w_3}/{k} )\le 0$ and the third equation of \eqref{eq:3} gives a contradiction. Therefore, $w_1=u_1+v_1<k$. A similar argument can be used to prove that $u_3+v_3>k$. 
\end{proof}

\begin{lemma}\label{lemma_nonexistence233}
    Suppose that $D$ and $Q$ are given by \eqref{DQ2}. Let $\bm r=(r, 0, 0)$ and $\bm s=(s_1, s_2, s_3)>\bm 0$ with $\bm s\neq \bm r$ and $\sum_{i=1}^3 s_i=r>0$.  Then model \eqref{patch} has no positive equilibrium. 
\end{lemma}
\begin{proof}
    Suppose to the contrary that $(\bm u, \bm v)$ is a positive equilibrium of \eqref{patch}. Since $r_2=r_3=0$, $(\bm u, \bm v)$ satisfies 
    \begin{equation}\label{eql2}
\begin{cases}
\ds -(d+q)u_1+du_2+ru_1\left(1-\frac{u_1+v_1}{k}\right)=0, \\
\ds (d+q)u_1-(2d+q)u_2+du_3=0, \\
(d+q)u_2-du_3=0.
\end{cases}
\end{equation}
Adding up the equations in \eqref{eql2}, we obtain $ru_1(1-(u_1+v_1)/k)=0$, which implies $u_1+v_1=k$. 
This contradicts Lemma  \ref{lemma_signc}.
\end{proof}

In the following two lemmas, we show that semitrivial equilibrium $E_2$ is unstable. 

\begin{lemma}\label{lemma_signd3}
    Suppose that $D$ and $Q$ are given by \eqref{DQ2}. Let  $\bm s=(s_1, s_2, s_3)\ge \bm 0$ with $s_2>0$ or $s_3>0$.  Then the semitrivial equilibrium $E_2=(\bm 0, \bm v^*)$ satisfies $v^*_1<k$.
\end{lemma}

\begin{proof}
We observe that $\bm v^*$ must satisfy
\[
0= \sum_{j=1}^3(dD_{ij}+qQ_{ij})v^*_j + s_iv_i^*\left(1-\frac{v_i^*}{k}\right), \ \ i=1, 2, 3.
\]
That is 
\begin{eqnarray}
&&0=-(d+q)v^*_1+dv^*_2+s_1v_1^*\left(1-\frac{v^*_1}{k} \right), \nonumber\\
&&0=(d+q)v^*_1-dv^*_2-(d+q)v^*_2+dv^*_3+s_2v_2^*\left(1-\frac{v^*_2}{k} \right), \label{eq:33}\\
&&0=(d+q)v^*_2-dv^*_3+s_3v_3^*\left(1-\frac{v^*_3}{k} \right). \nonumber
\end{eqnarray}
The rest proof is similar to that of Lemma \ref{lemma_signc}, so we omit it here.
\end{proof}

\begin{lemma}\label{lemma_unstable233}
    Suppose that $D$ and $Q$ are given by \eqref{DQ2}. Let $\bm r=(r, 0, 0)$ and $\bm s=(s_1, s_2, s_3)>\bm 0$ with $\bm s\neq \bm r$ and $\sum_{i=1}^3 s_i=r>0$.  Then the semitrivial equilibrium $E_2=(\bm 0, \bm v^*)$ is unstable  and the semitrivial equilibrium $E_1=(\bm u^*, 
    \bm 0)$ is stable for model \eqref{patch}. 
\end{lemma}
\begin{proof}
    The stability of $E_2$ is determined by the sign of  $\la_1(\bm{r}, \bm v^*)$, which is the principal eigenvalue of  $dD+qQ+\text{diag}(r_i(1-v^*_i/k))$.  
    By the assumptions on $\bm r$ and Lemma \ref{lemma_signd3}, we have 
$$
r_1\left(1-\frac{v^*_i}{k}\right)> 0
$$
and 
$$
r_i\left(1-\frac{v^*_i}{k}\right)=0, \ \ i=2, 3.
$$
Therefore, by Lemma \ref{lemma_sp}, we must have $\la_1(\bm{r}, \bm v^*)>s(dD+qQ)=0$. Hence, $E_2$ is  unstable.

 The stability of $E_1$ is determined by the sign of  $\la_1(\bm{s}, \bm u^*)$, which is the principal eigenvalue of  $dD+qQ+\text{diag}(s_i(1-u^*_i/k))$. By \cite{nguyen2022population}, we have $\bm u^*=(k,  (d+q)k/d, (d+q)^2k/d^2)$. So, 
$$
s_1\left(1-\frac{u^*_1}{k}\right)=0, 
$$
and 
$$
s_i\left(1-\frac{u^*_i}{k}\right)\le 0,  \ i=2, 3,
$$
with at least one strict sign by the assumption on $\bm s$. Therefore, by Lemma \ref{lemma_sp}, we must have $\la_1(\bm{s}, \bm u^*)<s(dD+qQ)=0$. Hence, $E_1$ is  stable.

\end{proof}

By the theory of monotone dynamical systems (see Lemma \ref{M1}) and Lemmas \ref{lemma_nonexistence233} and \ref{lemma_unstable233}, we obtain the following result:
\begin{theorem}
    Suppose that $D$ and $Q$ are given by \eqref{DQ2}. Let $\bm r=(r, 0, 0)$ and $\bm s=(s_1, s_2, s_3)>\bm 0$ with $\bm s\neq \bm r$ and $\sum_{i=1}^3 s_i=r>0$.  Then the semitrivial equilibrium $E_1=(\bm u^*, \bm 0)$ is globally asymptotically stable for model \eqref{patch}. 
\end{theorem}

{\large\noindent{\bf Declarations}}

\noindent{\bf Conflict of interest} The authors declare that they have no conflict of interest.

\bibliographystyle{abbrv}
\bibliography{ref}

\begin{thebibliography}{10}

\bibitem{bai2016optimization}
X.~Bai, X.~He, and F.~Li.
\newblock An optimization problem and its application in population dynamics.
\newblock {\em Proceedings of the American Mathematical Society},
  144(5):2161--2170, 2016.

\bibitem{berestycki2005analysis}
H.~Berestycki, F.~Hamel, and L.~Roques.
\newblock Analysis of the periodically fragmented environment model: I--species
  persistence.
\newblock {\em Journal of Mathematical Biology}, 51(1):75--113, 2005.

\bibitem{berman1994nonnegative}
A.~Berman and R.~J. Plemmons.
\newblock {\em Nonnegative matrices in the mathematical sciences}, volume~9 of
  {\em Classics in Applied Mathematics}.
\newblock Society for Industrial and Applied Mathematics (SIAM), Philadelphia,
  PA, 1994.

\bibitem{cantrell1989diffusive}
R.~S. Cantrell and C.~Cosner.
\newblock Diffusive logistic equations with indefinite weights: population
  models in disrupted environments.
\newblock {\em Proceedings of the Royal Society of Edinburgh Section A:
  Mathematics}, 112(3-4):293--318, 1989.

\bibitem{cantrell1991diffusive}
R.~S. Cantrell and C.~Cosner.
\newblock Diffusive logistic equations with indefinite weights: population
  models in disrupted environments {II}.
\newblock {\em SIAM Journal on Mathematical Analysis}, 22(4):1043--1064, 1991.

\bibitem{cantrell1998effects}
R.~S. Cantrell and C.~Cosner.
\newblock On the effects of spatial heterogeneity on the persistence of
  interacting species.
\newblock {\em Journal of Mathematical Biology}, 37(2):103--145, 1998.

\bibitem{cantrell2004spatial}
R.~S. Cantrell and C.~Cosner.
\newblock {\em Spatial Ecology via Reaction-Diffusion Equations}.
\newblock John Wiley \& Sons, 2004.

\bibitem{chen2022}
S.~Chen, J.~Liu, and Y.~Wu.
\newblock Evolution of dispersal in advective patchy environments with varying
  drift rates.
\newblock {\em Submitted}, 2022.

\bibitem{chen2022invasion}
S.~Chen, J.~Liu, and Y.~Wu.
\newblock Invasion analysis of a two-species {L}otka-{V}olterra competition
  model in an advective patchy environment.
\newblock {\em Studies in Applied Mathematics}, 149(3):762--797, 2022.

\bibitem{chen2021}
S.~Chen, J.~Shi, Z.~Shuai, and Y.~Wu.
\newblock Evolution of dispersal in advective patchy environments.
\newblock {\em Journal of Nonlinear Sciences}, 33:40(40):1--35, 2023.

\bibitem{cosner1996variability}
C.~Cosner.
\newblock Variability, vagueness and comparison methods for ecological models.
\newblock {\em Bull. Math. Biol.}, 58(2):207--246, 1996.

\bibitem{deangelis2016dispersal}
D.~L. DeAngelis, W.-M. Ni, and B.~Zhang.
\newblock Dispersal and spatial heterogeneity: single species.
\newblock {\em Journal of Mathematical Biology}, 72(1):239--254, 2016.

\bibitem{ding2010optimal}
W.~Ding, H.~Finotti, S.~Lenhart, Y.~Lou, and Q.~Ye.
\newblock Optimal control of growth coefficient on a steady-state population
  model.
\newblock {\em Nonlinear Analysis: Real World Applications}, 11(2):688--704,
  2010.

\bibitem{gourley2005two}
S.~A. Gourley and Y.~Kuang.
\newblock Two-species competition with high dispersal: the winning strategy.
\newblock {\em Mathematical Biosciences \& Engineering}, 2(2):345, 2005.

\bibitem{he2013effects}
X.~He and W.-M. Ni.
\newblock The effects of diffusion and spatial variation in lotka--volterra
  competition--diffusion system ii: the general case.
\newblock {\em Journal of Differential Equations}, 254(10):4088--4108, 2013.

\bibitem{he2016globalII}
X.~He and W.-M. Ni.
\newblock Global dynamics of the lotka--volterra competition--diffusion system
  with equal amount of total resources, ii.
\newblock {\em Calculus of Variations and Partial Differential Equations},
  55(2):25, 2016.

\bibitem{he2016globalI}
X.~He and W.-M. Ni.
\newblock Global dynamics of the lotka-volterra competition-diffusion system:
  diffusion and spatial heterogeneity i.
\newblock {\em Communications on Pure and Applied Mathematics},
  69(5):981--1014, 2016.

\bibitem{he2017global}
X.~He and W.-M. Ni.
\newblock Global dynamics of the lotka--volterra competition--diffusion system
  with equal amount of total resources, iii.
\newblock {\em Calculus of variations and partial differential equations},
  56(5):132, 2017.

\bibitem{hess}
P.~Hess.
\newblock {\em Periodic-Parabolic Boundary Value Problems and Positivity},
  volume 247 of {\em Pitman Research Notes in Mathematics Series}.
\newblock Longman Scientific \& Technical, Harlow, 1991.

\bibitem{hsu1996competitive}
S.~B. Hsu, H.~L. Smith, and P.~Waltman.
\newblock Competitive exclusion and coexistence for competitive systems on
  ordered {B}anach spaces.
\newblock {\em Trans. Amer. Math. Soc.}, 348(10):4083--4094, 1996.

\bibitem{hutson2003evolution}
V.~Hutson, S.~Martinez, K.~Mischaikow, and G.~T. Vickers.
\newblock The evolution of dispersal.
\newblock {\em J. Math. Biol.}, 47(6):483--517, 2003.

\bibitem{inoue2021unboundedness}
J.~Inoue and K.~Kuto.
\newblock On the unboundedness of the ratio of species and resources for the
  diffusive logistic equation.
\newblock {\em Discrete \& Continuous Dynamical Systems-Series B},
  26(5):2441--2450, 2021.

\bibitem{Jiang2020BMB}
H.~Jiang, K.-Y. Lam, and Y.~Lou.
\newblock Are {T}wo-{P}atch {M}odels {S}ufficient? {T}he {E}volution of
  {D}ispersal and {T}opology of {R}iver {N}etwork {M}odules.
\newblock {\em Bull. Math. Biol.}, 82(10):Paper No. 131, 42, 2020.

\bibitem{Jiang-Lam-Lou2021}
H.~Jiang, K.-Y. Lam, and Y.~Lou.
\newblock Three-patch models for the evolution of dispersal in advective
  environments: Varying drift and network topology.
\newblock {\em Bulletin of Mathematical Biology}, 83(10):1--46, 2021.

\bibitem{lam2016emergence}
K.-Y. Lam, Y.~Lou, and F.~Lutscher.
\newblock The emergence of range limits in advective environments.
\newblock {\em SIAM Journal on Applied Mathematics}, 76(2):641--662, 2016.

\bibitem{lam2016remark}
K.-Y. Lam and D.~Munther.
\newblock A remark on the global dynamics of competitive systems on ordered
  {B}anach spaces.
\newblock {\em Proc. Amer. Math. Soc.}, 144(3):1153--1159, 2016.

\bibitem{lamboley2016properties}
J.~Lamboley, A.~Laurain, G.~Nadin, and Y.~Privat.
\newblock Properties of optimizers of the principal eigenvalue with indefinite
  weight and robin conditions.
\newblock {\em Calculus of Variations and Partial Differential Equations},
  55(6):1--37, 2016.

\bibitem{li2010global}
M.~Y. Li and Z.~Shuai.
\newblock Global-stability problem for coupled systems of differential
  equations on networks.
\newblock {\em J. Differential Equations}, 248(1):1--20, 2010.

\bibitem{liang2012dependence}
S.~Liang and Y.~Lou.
\newblock On the dependence of population size upon random dispersal rate.
\newblock {\em Discrete \& Continuous Dynamical Systems-B}, 17(8):2771--2788,
  2012.

\bibitem{liang2021optimal}
X.~Liang and L.~Zhang.
\newblock The optimal distribution of resources and rate of migration
  maximizing the population size in logistic model with identical migration.
\newblock {\em Discrete \& Continuous Dynamical Systems-B}, 26(4):2055--2065,
  2021.

\bibitem{lin2014global}
K.-H. Lin, Y.~Lou, C.-W. Shih, and T.-H. Tsai.
\newblock Global dynamics for two-species competition in patchy environment.
\newblock {\em Mathematical Biosciences \& Engineering}, 11(4):947, 2014.

\bibitem{lou2006effects}
Y.~Lou.
\newblock On the effects of migration and spatial heterogeneity on single and
  multiple species.
\newblock {\em Journal of Differential Equations}, 223(2):400--426, 2006.

\bibitem{Lou14}
Y.~Lou and F.~Lutscher.
\newblock Evolution of dispersal in open advective environments.
\newblock {\em J. Math. Biol.}, 69(6-7):1319--1342, 2014.

\bibitem{lou2006minimization}
Y.~Lou and E.~Yanagida.
\newblock Minimization of the principal eigenvalue for an elliptic boundary
  value problem with indefinite weight, and applications to population
  dynamics.
\newblock {\em Japan Journal of Industrial and Applied Mathematics},
  23(3):275--292, 2006.

\bibitem{Lou15}
Y.~Lou and P.~Zhou.
\newblock Evolution of dispersal in advective homogeneous environment: the
  effect of boundary conditions.
\newblock {\em J. Differential Equations}, 259(1):141--171, 2015.

\bibitem{Lu1993}
Z.~Y. Lu and Y.~Takeuchi.
\newblock Global asymptotic behavior in single-species discrete diffusion
  systems.
\newblock {\em J. Math. Biol.}, 32(1):67--77, 1993.

\bibitem{Lutscher06}
F.~Lutscher, M.~A. Lewis, and E.~McCauley.
\newblock Effects of heterogeneity on spread and persistence in rivers.
\newblock {\em Bull. Math. Biol.}, 68(8):2129--2160, 2006.

\bibitem{mazari2019trait}
I.~Mazari.
\newblock Trait selection and rare mutations: The case of large diffusivities.
\newblock {\em Discrete and Continuous Dynamical Systems-Series B}, 2019.

\bibitem{mazari2020optimal}
I.~Mazari, G.~Nadin, and Y.~Privat.
\newblock Optimal location of resources maximizing the total population size in
  logistic models.
\newblock {\em Journal de math{\'e}matiques pures et appliqu{\'e}es},
  134:1--35, 2020.

\bibitem{mazari2022optimisation}
I.~Mazari, G.~Nadin, and Y.~Privat.
\newblock Optimisation of the total population size for logistic diffusive
  equations: bang-bang property and fragmentation rate.
\newblock {\em Communications in Partial Differential Equations},
  47(4):797--828, 2022.

\bibitem{mazari2021fragmentation}
I.~Mazari and D.~Ruiz-Balet.
\newblock A fragmentation phenomenon for a nonenergetic optimal control
  problem: Optimization of the total population size in logistic diffusive
  models.
\newblock {\em SIAM Journal on Applied Mathematics}, 81(1):153--172, 2021.

\bibitem{Nagahara2021}
K.~Nagahara, Y.~Lou, and E.~Yanagida.
\newblock Maximizing the total population with logistic growth in a patchy
  environment.
\newblock {\em Journal of Mathematical Biology}, 82(1):1--50, 2021.

\bibitem{nagahara2018maximization}
K.~Nagahara and E.~Yanagida.
\newblock Maximization of the total population in a reaction--diffusion model
  with logistic growth.
\newblock {\em Calculus of Variations and Partial Differential Equations},
  57(3):1--14, 2018.

\bibitem{nguyen2022population}
T.~D. Nguyen, Y.~Wu, A.~Veprauskas, T.~Tang, Y.~Zhou, C.~Beckford, B.~Chau,
  X.~Chen, B.~D. Rouhani, Y.~Wu, Y.~Yang, and Z.~Shuai.
\newblock Maximizing metapopulation growth rate and biomass in stream networks.
\newblock {\em arXiv preprint arXiv:2306.05555}, 2023.

\bibitem{pachepsky2005persistence}
E.~Pachepsky, F.~Lutscher, R.~Nisbet, and M.~A. Lewis.
\newblock Persistence, spread and the drift paradox.
\newblock {\em Theoretical Population Biology}, 67(1):61--73, 2005.

\bibitem{smith2008monotone}
H.~L. Smith.
\newblock {\em Monotone {D}ynamical {S}ystems: {A}n {I}ntroduction to the
  {T}heory of {C}ompetitive and {C}ooperative {S}ystems}.
\newblock American Mathematical Society, Providence, RI, 1995.

\bibitem{speirs2001population}
D.~C. Speirs and W.~S.~C. Gurney.
\newblock Population persistence in rivers and estuaries.
\newblock {\em Ecology}, 82(5):1219--1237, 2001.

\bibitem{takeuchi1996global}
Y.~Takeuchi.
\newblock {\em Global dynamical properties of {L}otka-{V}olterra systems}.
\newblock World Scientific, 1996.

\bibitem{vasilyeva2012flow}
O.~Vasilyeva and F.~Lutscher.
\newblock How flow speed alters competitive outcome in advective environments.
\newblock {\em Bull. Math. Biol.}, 74(12):2935--2958, 2012.

\bibitem{wang2019persistence}
Y.~Wang, J.~Shi, and J.~Wang.
\newblock Persistence and extinction of population in
  reaction--diffusion--advection model with strong allee effect growth.
\newblock {\em Journal of mathematical biology}, 78(7):2093--2140, 2019.

\bibitem{wei2021coexistence}
J.~Wei and B.~Liu.
\newblock Coexistence in a competition-diffusion-advection system with equal
  amount of total resources.
\newblock {\em Mathematical Biosciences and Engineering}, 18(4):3543--3558,
  2021.

\bibitem{yan2022competition}
X.~Yan, H.~Nie, and P.~Zhou.
\newblock On a competition-diffusion-advection system from river ecology:
  mathematical analysis and numerical study.
\newblock {\em SIAM Journal on Applied Dynamical Systems}, 21(1):438--469,
  2022.

\bibitem{zhang2017carrying}
B.~Zhang, A.~Kula, K.~M. Mack, L.~Zhai, A.~L. Ryce, W.-M. Ni, D.~L. DeAngelis,
  and J.~D. Van~Dyken.
\newblock Carrying capacity in a heterogeneous environment with habitat
  connectivity.
\newblock {\em Ecology Letters}, 20(9):1118--1128, 2017.

\bibitem{zhang2015effects}
B.~Zhang, X.~Liu, D.~L. DeAngelis, W.-M. Ni, and G.~G. Wang.
\newblock Effects of dispersal on total biomass in a patchy, heterogeneous
  system: Analysis and experiment.
\newblock {\em Mathematical Biosciences}, 264:54--62, 2015.

\bibitem{MR4222368}
P.~Zhou, D.~Tang, and D.~Xiao.
\newblock On {L}otka-{V}olterra competitive parabolic systems: exclusion,
  coexistence and bistability.
\newblock {\em J. Differential Equations}, 282:596--625, 2021.

\end{thebibliography}

\end{document}